\documentclass[a4paper,UKenglish,cleveref, autoref, thm-restate]{lipics-v2019}

\usepackage{microtype}
\usepackage{float}
\graphicspath{{./images/}}
\usepackage{tikz}
\usepackage{xspace}
\usepackage{url}
\usepackage{afterpage}

\newcommand{\pspace}{\textsc{PSPACE}\xspace}
\newcommand{\np}{\textsc{NP}\xspace}
\renewcommand{\S}{\ensuremath{\mathcal{S}}\xspace}
\newcommand{\T}{\ensuremath{\mathcal{T}}\xspace}
\newcommand{\andnode}{{\sf\emph{AND}}\xspace}
\newcommand{\ornode}{{\sf\emph{protected OR}}\xspace}
\newcommand{\rin}{\ensuremath{\mathit{in}}\xspace}
\newcommand{\rout}{\ensuremath{\mathit{out}}\xspace}

\newtheorem{observation}[theorem]{Observation}

\title{Multi-robot motion planning of $k$-colored discs is PSPACE-hard}

\author{Thomas Brocken}{TU Eindhoven, the Netherlands}{thomasbrocken@live.nl}{}{}
\author{G. Wessel van der Heijden}{TU Eindhoven, the Netherlands}{g.w.v.d.heijden@student.tue.nl}{}{}
\author{Irina Kostitsyna}{TU Eindhoven, the Netherlands}{i.kostitsyna@tue.nl}{}{}
\author{Lloyd E. Lo-Wong}{TU Eindhoven, the Netherlands}{l.e.lo-wong@student.tue.nl}{}{}
\author{Remco J. A. Surtel}{TU Eindhoven, the Netherlands}{remcosurtel@hotmail.com}{}{}

\authorrunning{T. Brocken, G. W. van der Heijden, I. Kostitsyna, L. E. Lo-Wong, and R. J. A. Surtel} 

\Copyright{Thomas Brocken, G. Wessel van der Heijden, Irina Kostitsyna, Lloyd E. Lo-Wong and Remco J. A. Surtel} 



\nolinenumbers 

\EventEditors{Martin Farach-Colton, Giuseppe Prencipe, and Ryuhei Uehara}
\EventNoEds{3}
\EventLongTitle{10th International Conference on Fun with Algorithms (FUN 2020)}
\EventShortTitle{FUN 2020}
\EventAcronym{FUN}
\EventYear{2020}
\EventDate{September 28--30, 2020}
\EventLocation{Favignana Island, Sicily, Italy}
\EventLogo{}
\SeriesVolume{157}
\ArticleNo{15}
\hideLIPIcs
\begin{document}

\maketitle

\begin{abstract}
In the problem of \emph{multi-robot motion planning}, a group of robots, placed in a polygonal domain with obstacles, must be moved from their starting positions to a set of target positions.
We consider the specific case of unlabeled disc robots of two different sizes.
That is, within one class of robots, where a class is given by the robots' size, any robot can be moved to any of the corresponding target positions.
We prove that the decision problem of whether there exists a schedule moving the robots to the target positions is \pspace-hard.
\end{abstract}

\section{Introduction}
Due to a wide range of applications, the multi-robot motion planning problem has received a great amount of attention in the theoretical computer science community in recent years.
In the most general setting, the problem can be phrased in the following way: given a set of robots placed in a polygonal domain, find a schedule to move the robots from their initial locations to some specified target locations without collisions.
From the point of view of identifying which robots move to which target positions, we can distinguish between \emph{labeled} and \emph{unlabeled} robot motion planning.
Labeled motion planning is the most studied and, possibly, is a more natural variant of the problem.
In it the robots have unique IDs, and each robot has a specifically assigned target location.
In this paper, however, we are more interested in unlabeled robot motion planning, where the robots are indistinguishable from one another, and each robot can move to any of the specified target locations.
A classic example of a motivating application for this problem is a swarm of robots operating in a warehouse, where it does not matter which of the robots arrives to pick up an item to be transported.
Generalizing the notions of labeled and unlabeled motion planning, Solovey and Halperin~\cite{Solovey2014} introduce the \emph{$k$-color} robot motion planning problem, where given are $k$ classes of robots and $k$ sets of target positions.
Within each class the robots are unlabeled, and each robot may move to any location in the corresponding set of target positions.
When $k=n$ the problem becomes the standard labeled version of the robot motion planning, and when $k=1$ it is the unlabeled version.

In this paper we consider the $k$-color Disc-Robot Motion Planning problem, $k$-DRMP, where classes of robots differ only by their radii.
We show that the problem of deciding whether a particular target location can be reached by a robot from the corresponding class is \pspace-hard.
Our results imply that a version of the Sliding Block game with round pieces can make for a fun and interesting puzzle.

\subparagraph{Related work.}
We start with a brief overview of the known algorithmic results for the disc robot motion planning problem.
For unlabeled unit disc robots inside a simple polygon, Adler et al.~\cite{Adler2015} develop a polynomial-time algorithm to solve the problem under an additional requirement that the distance between any two points from the union of the starting and target locations is at least $4$.
For unlabeled unit disc robots inside a polygonal domain with obstacles, Solovey et al.~\cite{Solovey2015} show how to find a solution close to optimal in polynomial time.
In addition to the same requirement on the separation between the starting/target positions, they require the minimum distance between a robot location and an obstacle to be at least $\sqrt{5}$.

In contrast, for a set of disc robots of possibly different radii in a simple polygon, it is \np-hard to decide whether a target location can be reached by any robot~\cite{spirakis1984strong}.

A wider range of hardness results exists for rectangular or square-shaped robots.
Many of these are inspired by \emph{Sliding Block} puzzles, a family of popular games where different shapes are densely packed in a rectangular grid box with little room for movement, and the goal is to free a specific target block and move it outside of the box by sliding the pieces around. Figure~\ref{fig:sliding-blocks} shows an example of a puzzle where the blocks are rectangles of integer side length.

\begin{figure}[t]
\centering
\includegraphics[page=1]{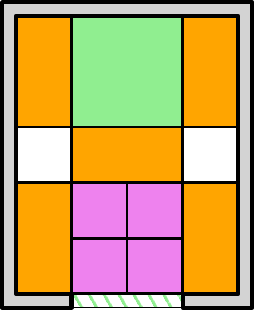}
\caption{An example of a Sliding Block puzzle. The goal is to take the $2\times 2$ green square outside the box through the exit on the bottom (through which only the green square can slide).}
\label{fig:sliding-blocks}
\end{figure}

One of the earliest results is due to Hopcroft et al.~\cite{hopcroft1984complexity}, where they show that it is \pspace-hard to decide whether a given set of rectangular robots enclosed in a rectangular domain can be reconfigured into a particular target configuration.
Flake and Baum~\cite{Flake2002} showed that solving the \emph{Rush Hour} puzzle on an $n\times n$ grid is \pspace-complete.
Rush Hour is a type of a sliding-block puzzle, where the blocks, called cars, are rectangles of width $1$ and of length either $2$ or $3$, and the cars are only allowed to move parallel to their longer side.
To prove the hardness of the Rush Hour puzzle, Flake and Baum develop a new specialized model of computation based on ``Generalized Rush Hour'' logic.
They show how to simulate a Finite Turing Machine with circuits built on this logic, which settles the complexity of the problem.

Inspired by Flake and Baum's construction, Hearn and Demaine \cite{hearn2005pspace} develop a \emph{Nondeterministic Constraint Logic} (NCL) framework which has proven to be invaluable in showing hardness results for many problems, based on puzzles and otherwise.
To showcase the power of the NCL, they use it to prove \pspace-completeness of a number of puzzles, including the Sliding Block, even when all blocks are small (in particular, of size $1\times 2$), the classic Rush Hour (for cars of length $2$ and $3$), and others.
Tromp and Cilibrasi~\cite{Tromp2005} used the NCL framework to show that Rush Hour is \pspace-complete even for the cars of length $2$ alone.
Finally, using NCL, Solovey and Halperin~\cite{solovey2016hardness} prove that unlabeled multi-robot motion planning for unit square robots moving amidst polygonal obstacles is \pspace-hard.

\subparagraph{Contribution.}
The contrast between the abundance of hardness results for rectangular and square robots and a few results, negative as well as positive, for disc robots, suggests that the complexity of the problem greatly depends on the shape of the robots, even when the difference between the shapes is seemingly insignificant.
Establishing the complexity of multi-robot motion planning of unit disc robots has been an open problem for quite some time.
In this paper we show the first \pspace-hardness result for motion planning of disc robots.
In particular, we show that the $2$-color multi-robot motion planning problem for disc robots with radii $1/2$ or $1$ in a polygonal domain is \pspace-hard by a reduction from the NCL.
In contrast, the \np-hardness construction of~\cite{spirakis1984strong} uses discs of very different sizes with a large ratio between the largest and the smallest disc.

The rest of the paper is structured in the following way.
In Section~\ref{sec:prelim} we introduce a formal problem statement and overview the NCL.
In Section~\ref{sec:reduction} we show the hardness reduction.
We start with describing the gadgets in Section~\ref{sec:gadgets}, and prove their correctness in Section~\ref{sec:correctness}.
Finally, in Section~\ref{sec:main} we state our main results.

\section{Problem statement and preliminaries}
\label{sec:prelim}

In this section we start with a few definitions, and we state the $k$-color Disc-Robot Motion Planning ($k$-DRMP) problem more formally.

Let $P$ be a polygonal domain in the plane. By $D(p,r)$ we denote a disc of radius $r$ centered at a point $p$. A point $p\in P$ is a \emph{valid} position for a disc robot with radius $r > 0$, if $D(p,r)$ is fully contained in $P$. A set of points $S=\{p_1,\dots,p_n\}\subset P$ is a \emph{valid configuration} for a set of robots with radius $r$ if (1)~$p_i$ is a valid position for a robot of radius $r$ for all $1\leq i \leq n$, and (2)~discs $D(p_i,r)$ and $D(p_j,r)$ do not intersect in their interior, that is, if $|p_i-p_j|\ge 2r$, for all $1\leq i<j \leq n$.

For $k$ distinct positive radii $\{r_1,r_2,\dots,r_k\}$ and $k$ positive integers $\{n_1,n_2,\dots,n_k\}$, denote a \emph{$k$-configuration} to be a set of $k$ configuration-radius pairs $\S=\{(S_i,r_i): |S_i|=n_i\}$. A $k$-configuration is \emph{valid} if each $S_i$ is a valid configuration for disc robots of radius $r_i$, and for all $i<j$, all $x\in S_i$ and all $y\in S_j$, the discs of respective radii centered at $x$ and $y$ do not intersect in their interior, that is, $|x-y|\ge r_i+r_j$. We will refer to the set of robots with the same radius as a \emph{class} of robots. Thus, each $S_i$ specifies a configuration of a class of robots with radius $r_i$.

We say that a set of $n$ disc robots with radius $r$ can be \emph{reconfigured} from a valid configuration $S$ into a valid configuration $T$, if $|S|=|T|=n$, and there exist $n$ paths $\{\pi_1,\pi_2,\dots,\pi_n\}$, where each path $\pi_i:[0,1]\rightarrow\mathbb{R}^2$ is a continuous curve, such that their starting points form the set $S$ (i.e., $\bigcup\limits_i \pi_i(0)=S$), their final points form the set $T$, (i.e., $\bigcup\limits_i \pi_i(1)=T$), and at any moment in time $t\in[0,1]$, the set of points $\{\pi_1(t),\pi_2(t),\dots,\pi_n(t)\}$ forms a valid configuration for the given value of $r$.

Analogously, we say that $k$ classes of robots can be \emph{reconfigured} from a valid $k$-configuration $\S=\{S_1,S_2,\dots,S_k\}$ into a valid $k$-configuration $\T=\{T_1,T_2,\dots,T_k\}$, if $|S_i|=|T_i|$ for all $i$, and there exist a set of paths which reconfigure each $S_i$ into $T_i$, such that no two robots overlap at any moment in time.


Drawing inspiration from Hearn and Demaine~\cite{hearn2005pspace} and Solovey and Halperin~\cite{solovey2016hardness}, we define a few variants of the $k$-DRMP problem.
\begin{description}
\item[Multi-to-multi $k$-DRMP] Given $k$ classes of robots and two valid $k$-configurations $\S$ and $\T$, decide whether the robots can be reconfigured from $\S$ to $\T$.
\item[Multi-to-single $k$-DRMP] Given $k$ classes of robots, a valid $k$-configuration $\S$, and a target position $t\in P$, decide whether there exists a valid $k$-configuration $\T$ with $t\in T_i$ for some $T_i \in \T$, such that the robots can be reconfigured from $\S$ to $\T$.
\item[Multi-to-single-in-class $k$-DRMP] Given $k$ classes of robots with distinct radii $\{r_1,\dots,r_k\}$, some $1\le i \le k$, a valid $k$-configuration $\S$, and a target position $t\in P$, decide whether there exists a valid $k$-configuration $\T$ with $t\in T_i$, where $T_i \in \T$ is the target configuration for the robots with radius $r_i$, such that the robots can be reconfigured from $\S$ to $\T$.
\end{description}

Intuitively, the \emph{multi-to-multi} problem can be interpreted as follows: Let $k$-configurations $\S$ and $\T$ represent the start and target positions respectively for the $k$ classes of robots.
Can the robots move from $\S$ to $\T$ without any collisions?
In this paper we will prove that this, and the other two variants of the $k$-DRMP problem, are \pspace-hard.
Note, that it is possible to define more variants of the $k$-DRMP problem by varying which starting or target positions might be fixed, possibly with a fixed matching on them, with specified robot radii, or in any other way along these lines.
Many of them can be shown \pspace-hard with a slight modification to our reduction.

\subsection{Nondeterministic constraint logic}
We will now briefly introduce the \emph{nondeterministic constraint logic} (NCL).
Hearn and Demaine~\cite{hearn2005pspace} define an NCL machine as a weighted graph $G = (V, E)$, with nonnegative integer weights on the edges, and with integer \emph{minimum in-flow} constraints on the nodes.
A \emph{state} of the NCL machine is an assignment of directions onto the edges of $G$.
A state is \emph{valid} if for every node the total weight of incoming edges is at least the value of the minimum in-flow constraint of that node.


%


Consider a valid state of the NCL machine, and some edge $e=(u,v)\in E$ directed from $u$ to $v$.
We can perform an \emph{edge flip} by reassigning the orientation of $e$ from $v$ to $u$, as long as the state after the flip remains valid, that is, the in-flow constraint of $u$ is still satisfied.
The edge flip operation describes possible transitions between the states of the NCL machine.

Hearn and Demaine~\cite{hearn2005pspace} show that, even for very restricted versions of the NCL machine, it is \pspace-complete to decide whether there exists a sequence of valid edge flips which transforms one valid state into another.
In particular, the \pspace-completeness holds for the following four decision problems on an NCL machine which is (1) defined on a simple planar graph $G=(V,E)$, (2) has edge weights either $1$ or $2$, (3) has the minimum in-flow constraint $2$ on all nodes, and (4) has nodes of types \andnode or \ornode (which we describe later). The decision problems are:
\begin{description}
    \item[State-to-state] Given two states $\sigma_1$ and $\sigma_2$, decide whether there exists a valid sequence of edge flips that transforms $\sigma_1$ into $\sigma_2$.
    \item[State-to-edge] Given a state $\sigma_1$ and an edge $e \in E$, decide whether there exists a state $\sigma_2$ such that $e$ has the opposite orientations in $\sigma_1$ and $\sigma_2$, and there exists a valid sequence of edge flips that transforms $\sigma_1$ into $\sigma_2$.
    \item[Edge-to-edge] Given two edges $e_1$ and $e_2$ with specific orientations, decide whether there exist two states $\sigma_1$ and $\sigma_2$ with $e_1$ and $e_2$ of prescribed orientation respectively, and there exists a valid sequence of edge flips that transforms $\sigma_1$ into $\sigma_2$.
    \item[Edge-to-state] (symmetric to state-to-edge) Given an edge $e \in E$ and a state $\sigma_2$, decide whether there exists a state $\sigma_1$ such that $e$ has the opposite orientations in $\sigma_1$ and $\sigma_2$, and there exists a valid sequence of edge flips that transforms $\sigma_1$ into $\sigma_2$.
\end{description}

The two types of nodes, the \andnode and the \ornode, are both degree three nodes with the following properties (refer to Figure~\ref{fig:nodes}). The \andnode node has two incident edges of weight $1$, and one incident edge of weight $2$. Thus, to satisfy the in-flow constraint, the weight-$2$ edge can be directed outwards only if both weight-$1$ edges are directed inwards. All incident edges of an \textsf{\emph{OR}} node have weight $2$. Thus, an edge can be directed outwards if at least one other incident edges is directed inwards. In a \ornode node two incident edges are labeled as \emph{input}, and one as \emph{output}. The `protected' property forbids the two input edges to be directed inwards at the same time. In many cases, including ours, this restriction simplifies reductions.


\begin{figure}[t]
\centering
\begin{minipage}{.4\textwidth}
\centering
\includegraphics[page=1]{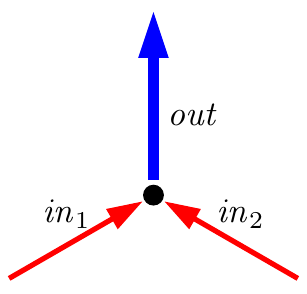}
\end{minipage}
\hfil
\begin{minipage}{.4\textwidth}
\centering
\includegraphics[page=2]{nodes.pdf}
\end{minipage}
\caption{Two types of nodes.
Edges with weight 1 are shown in red and edges with weight 2 are shown in blue. Left: An \andnode node. The minimum in-flow requirement of the vertex is 2.  Edge $\rout$ can be directed outwards only if both $\rin_1$ and $\rin_2$ are directed inwards, if the minimum in-flow constraint were to be maintained. Right: A \ornode node. Edge $\rout$ can be directed outwards if either $\rin_1$ and $\rin_2$ is directed inwards. In a \ornode it is not possible for both in-edges to be directed inwards simultaneously.}
\label{fig:nodes}
\end{figure}

\begin{theorem}[Theorem 11~\cite{hearn2005pspace}]
State-to-edge, edge-to-state, state-to-state, and edge-to-edge are \pspace-complete, even when the constraint graph is simple, planar, and only has nodes of types \andnode and \ornode.
\end{theorem}



\section{From NCL to $2$-DRMP}
\label{sec:reduction}

In this section we will prove that the $k$-DRMP problem is \pspace-hard for $k$ classes of unlabeled disc robots moving amidst obstacles constructed out of line segments and circular arcs, even if $k=2$.
We reduce from the NCL problem, for a given constraint graph $G$ we construct a $2$-DRMP instance that emulates the NCL machine built on $G$.
Then, by considering the \emph{state-to-edge} and \emph{state-to-state} versions of the NCL problem we will show that the three variants of the $2$-DRMP problem defined in Section~\ref{sec:prelim} are \pspace-hard.
%

\begin{figure}[t]
\begin{minipage}[t]{0.48\textwidth}
\centering
\includegraphics{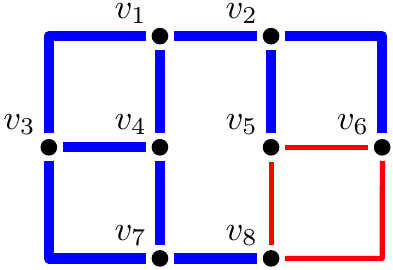}
\caption{Orthogonal drawing of a constraint graph $G$. Edges with weight $1$ are shown in red, while edges with weight $2$ are shown in blue. Nodes $v_5$, $v_6$, and $v_8$ are \andnode-nodes, and nodes $v_1$, $v_2$, $v_3$, $v_4$ and $v_7$ are \ornode-nodes.}
\label{fig:grids}
\end{minipage}
\hfill
\begin{minipage}[t]{0.48\textwidth}
\centering
    \includegraphics{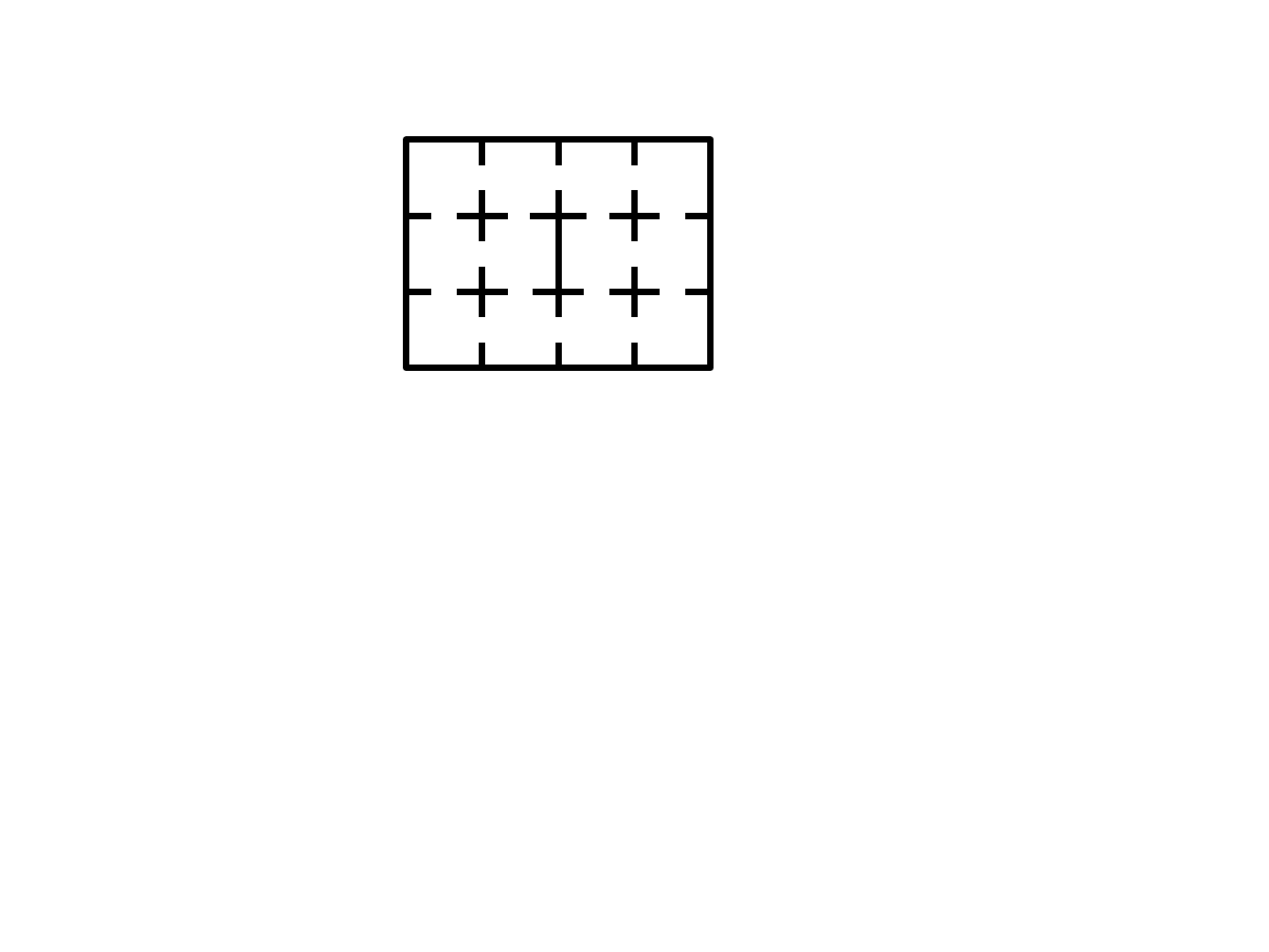}
    \caption{The grid-environment to be filled with gadgets, corresponding to graph $G$ shown in Figure~\ref{fig:grids}. Square cells either correspond to the nodes of $G$, or grid points that edges pass through.}
    \label{fig:gadget_grid}
\end{minipage}
\end{figure}

Consider an instance of a constraint graph $G=(V,E)$ with \andnode- and \ornode-nodes, and consider its orthogonal drawing on a square grid, with the nodes having integer coordinates, and edges having at most one bend~\cite{Calamoneri1995} (refer to Figure~\ref{fig:grids} for an example). We will construct an instance of the $2$-DRMP problem such that it is equivalent to the NCL problem on $G$.

Similarly to the constructions in~\cite{Flake2002,hearn2005pspace,solovey2016hardness}, we create a grid-like environment (refer to Figure~\ref{fig:gadget_grid}): square cells of size $18\times 18$ are separated with walls of width $1$.
The cells correspond to the nodes of $G$ or to the edges passing through grid points.
If there is an edge in $G$ passing between two adjacent cells, there is an opening of width $2$ in the middle of the wall separating the two cells.
In each cell we construct free space (by filling its complement with  obstacles) and densely place two classes of disc robots with radii $1/2$ and $1$ to emulate the nodes and the edges of $G$.
The details of the three types of gadgets are described in the following section.


\subsection{Gadgets}\label{sec:gadgets}
The three types of gadgets needed to emulate an NCL machine are \andnode gadgets, \ornode gadgets, and connector gadgets.
Before we describe the specifics of each of the gadgets, we introduce a few building components which will appear in all of them.

\subparagraph{Edge robot.} Every gadget has two or three length-$2$ openings in the middle of the wall edges surrounding the square cell of the gadget.
Next to each opening a radius-$1$ \emph{edge} robot is placed (shown in green in the figures), which is shared by the two adjacent gadgets, and which may go back and forth through the opening.
The construction of the gadgets is such that only the edge robots can enter and leave their corresponding gadgets.
The rest of the robots, called \emph{internal} robots, will always remain within their gadget.
Edge robots are restricted in their movement by obstacles either as depicted in Figure~\ref{fig:terminal_configurations} or by an obstacle as depicted in Figure~\ref{fig:nudge_edge_robot}.
Consider an opening between two gadgets corresponding to some edge $(u,v)\in E$.
Denote by $m_u$ and $m_v$ the two midpoints on the longer edges of the $1\times 2$ rectangle forming the opening, respectively closer to the gadgets corresponding to $u$ and $v$ (refer to Figure~\ref{fig:terminal_configurations}).
We will call these points \emph{terminal} positions of the edge robot.
With respect to a given gadget, we will distinguish between an \emph{inside} terminal position and an \emph{outside} terminal position of an edge robot.
\begin{figure}[t]
    \centering
    \includegraphics{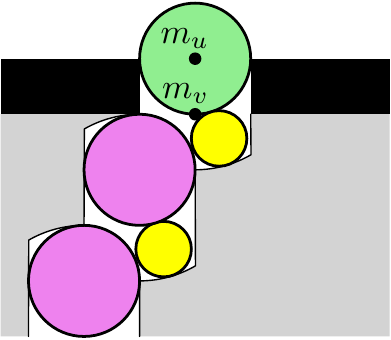}
\hfil
    \includegraphics{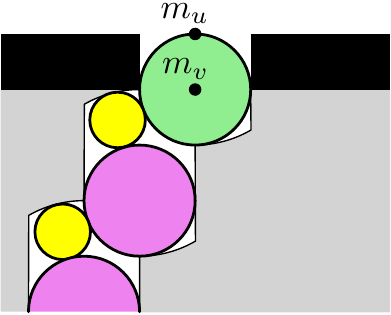} 
\caption{Two terminal positions of an edge robot. Left: edge directed from $u$ to $v$. Right: edge directed from $v$ to $u$.}
    \label{fig:terminal_configurations}
\end{figure}
\begin{figure}[t]
    \centering
    \includegraphics{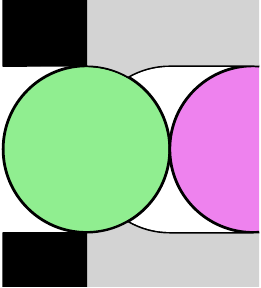}
    \caption{Nudges at an edge robot. The slight narrowing in the free space does not let the interior violet disc to travel more than a unit distance to the left.}
    \label{fig:nudge_edge_robot}
\end{figure}%
The terminal positions of the edge robot will correspond to the orientation of the edge $(u,v)$ in $G$: the robot centered at $m_u$ corresponds to the edge directed from $u$ to $v$ (Figure~\ref{fig:terminal_configurations} (left)), and the robot centered at $m_v$ corresponds to the edge directed from $v$ to $u$ (Figure~\ref{fig:terminal_configurations} (right)).
Intermediate positions of the robot do not define a specific orientation of the edge, and may correspond to any direction.
Thus, the positions of the edge robots, if they are all in terminal positions, will fully describe a state of an NCL machine.

\subparagraph{Parallel and perpendicular components.}
The two constructions shown in Figures~\ref{fig:par-comp} and~\ref{fig:perp-comp} are called a \emph{parallel component} and a \emph{perpendicular component} respectively.
The parallel component consists of free space formed by two parallel $2\times 3$ rectangles (outlined with a dashed line) overlapping in a corner $1\times 1$ square.
Place two radius-$1$ discs (violet in the figure) and a radius-$(1/2)$ disc (yellow in the figure) in the resulting free space.
For each $2\times 3$ rectangle consider two points placed at a unit distance from three of the four sides of the rectangle.
These points are the terminal positions for the two radius-$1$ discs.
Finally, the short unit-length boundary edges of the free space are rounded (replaced with radius-$2$ arcs centered at the further terminal positions) so that the radius-$(1/2)$ disc touches a radius-$1$ disc when moving into a corner of a rectangle.

The perpendicular component consists of two perpendicular $2\times 3$ rectangles (in dashed) overlapping in a corner $1\times 1$ square.
Similarly, place two radius-$1$ discs and a radius-$(1/2)$ disc in the free space, and consider the four terminal positions for the radius-$1$ discs.
Again, the short unit length boundary edges are replaced with radius-$2$ arcs centered at the further terminal positions.

\begin{figure}[t]
\centering
\begin{minipage}[t]{0.45\textwidth}
\centering
    \includegraphics[page=1]{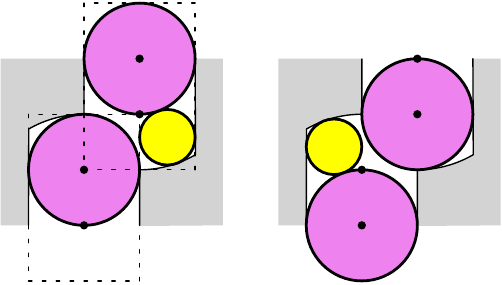}
    \caption{Two configurations of a parallel component.}
    \label{fig:par-comp}
\end{minipage}
\hfill
\begin{minipage}[t]{0.45\textwidth}
\centering
    \includegraphics[page=2]{images/par-perp-comp.pdf}
    \caption{Two configurations of a perpendicular component.}
    \label{fig:perp-comp}
\end{minipage}
\end{figure}

\begin{figure}[t]
\begin{minipage}{\textwidth}
    \centering
    \includegraphics[scale=0.9,page=2]{images/gadgets/2-color/AND_1.pdf}
    \hfil
    \includegraphics[scale=0.9]{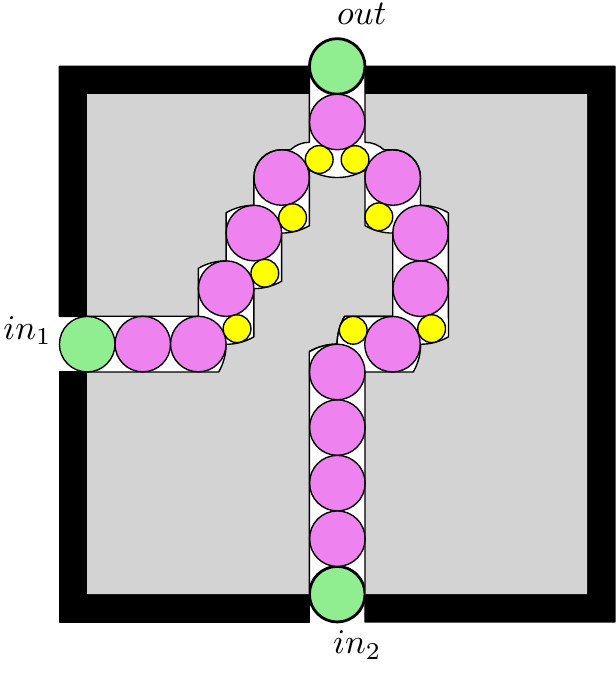}
    \caption{\andnode-gadgets representing \andnode nodes. Edge robot $\rout$ can move inside if and only if both edge robots $\rin_1$ and $\rin_2$ are outside.}
    \label{fig:and}
\end{minipage}
\begin{minipage}{\textwidth}
    \centering
    \includegraphics[scale=0.9,page=2]{gadgets/OR_FINAL.pdf}
    \hfil
    \includegraphics[scale=0.9]{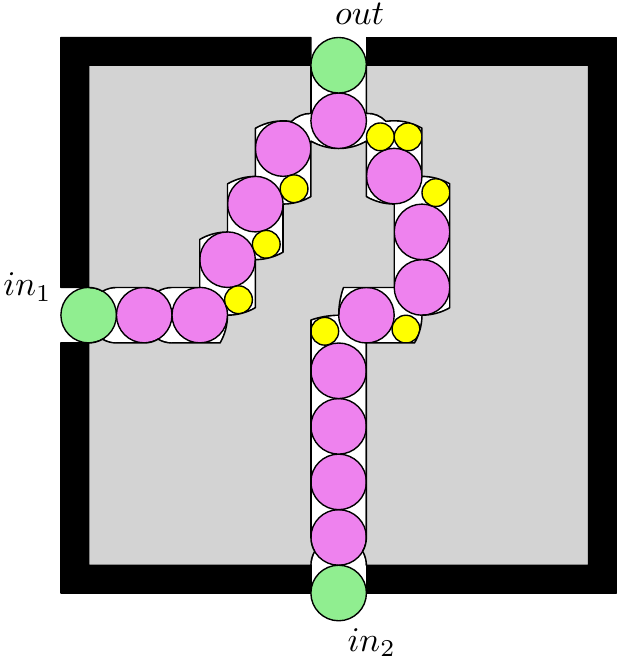}
    \caption{The \ornode-gadgets representing \ornode vertices. Robot $out$ can move inside if and only if either $in_1$ or $in_2$ is outside. In this figure, $in_2$ has moved out allowing robot $out$ to move in.}
    \label{fig:or}
\end{minipage}
\begin{minipage}{\textwidth}
    \centering
        \includegraphics[scale=0.9]{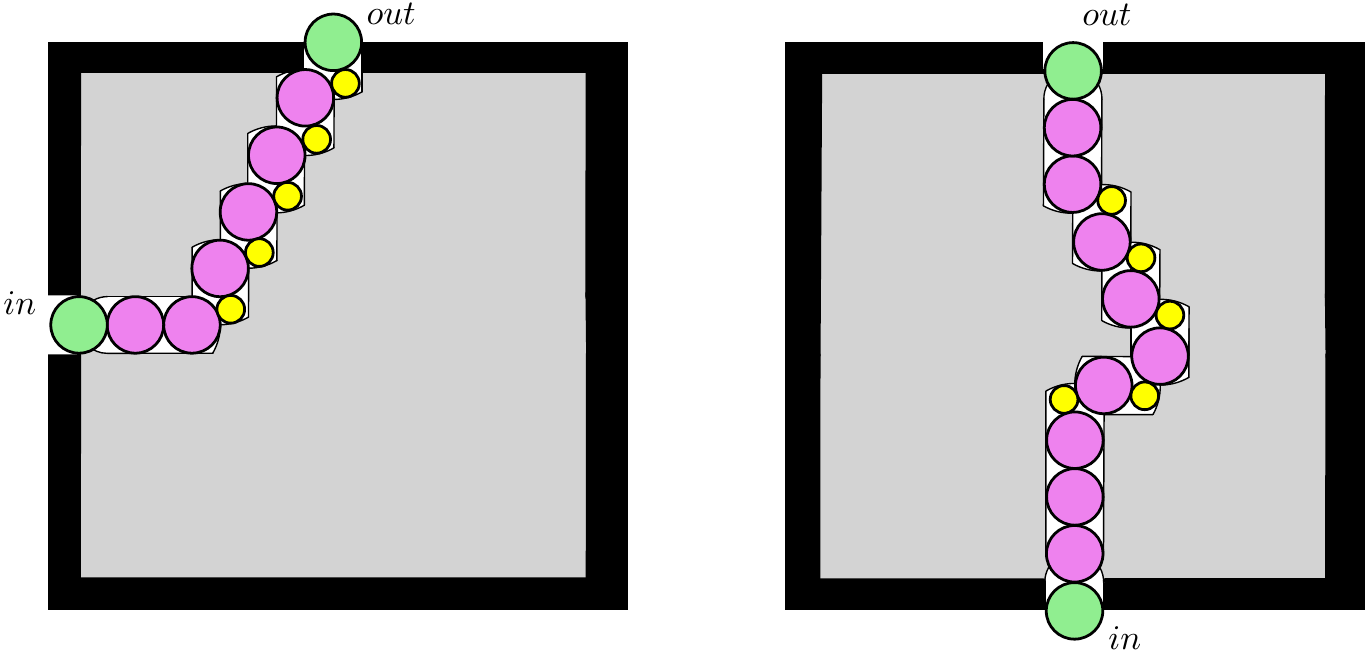}
    \caption{Connector gadgets representing connector vertices. Robot $\rout$ can move inside only if robot $\rin$ is outside and vice versa.}
    \label{fig:connector_gadgets}
\end{minipage}
\end{figure}

\afterpage{\clearpage}

The parallel and perpendicular components are designed for propagating a signal of the position of an edge robot.
In Section~\ref{sec:correctness} we will show the following property of a chain of parallel and perpendicular configurations.
Let $\mathcal{C}$ be a chain of at least $1$ parallel and $1$ perpendicular configuration, possibly extended with aligned radius-$1$ discs (for example, as the one used in the connector gadget shown in Figure~\ref{fig:connector_gadgets} (right)).
Let $A^{*}$ be the first radius-$1$ robot in $\mathcal{C}$, and $B^{*}$ be the last radius-$1$ robot in $\mathcal{C}$.
Let $t_1$ and $t_2$ be the terminal positions of $A^*$, and $t_3$ and $t_4$ be the terminal positions of $B^*$, such that $t_1$, $t_2$, $t_3$, and $t_4$ appear in order along $\mathcal C$.
Then the following lemma holds.
\begin{lemma}\label{lem:signal}
Robot $A^*$ can move to its terminal position $t_2$ only if $B^*$ is in its terminal position $t_4$.
Similarly, $B^*$ can move to its terminal position $t_3$ only if $A^*$ is in its terminal position $t_1$.
\end{lemma}
Using this lemma we will use a sequence of at least $2$ parallel and perpendicular components to force the edge robots to be located at one of their terminal positions. We are now ready to show the construction of the gadgets.

\subparagraph{\andnode gadgets.}
There are two versions of the \andnode gadget (up to mirror symmetry and rotation by $90^\circ$, $180^\circ$, or $270^\circ$), shown in Figure~\ref{fig:and}.
Also refer to Figure~\ref{fig:AND_closeup}.
They are designed in such a way that the edge robot $\rout$ can only be moved to the inside terminal position if edge robots $\rin_1$ and $\rin_2$ are both moved to their outside terminal positions.
Indeed, the edge robot $\rout$ can move to the inside terminal position only if the two yellow discs of radius $1/2$ below it move to the left and to the right, respectively.
These yellow discs can move left and right only if both radius-$1$ violet discs touching them move down, which by Lemma~\ref{lem:signal} can happen only if both edge robots $\rin_1$ and $\rin_2$ move to their outside terminal positions.
In Section~\ref{sec:correctness} we will prove the following lemma.
\begin{lemma}\label{lem:and-correct}
In the \andnode gadgets, the edge robot `$\rout$' can be moved to the inside terminal position only if the edge robots `$\rin_1$' and `$\rin_2$' are both moved to (or beyond) their outside terminal positions.
The edge robot `$\rin_1$' (`$\rin_2$') can be moved to its inside terminal position only if the edge robot `$\rout$' is moved to (or beyond) its outside terminal position.
\end{lemma}


\subparagraph{{\sf\emph{Protected OR}} gadgets.} Similarly to the \andnode gadgets, there are two versions of the \ornode gadget shown in Figure~\ref{fig:or}.
Also refer to Figure~\ref{fig:OR_closeup}.
The \ornode gadgets are designed in such a way that the edge robot $\rout$ can move to the inside terminal position of the gadget if and only if either $\rin_1$ or $\rin_2$ move to their outside terminal positions.
Indeed, the edge disc $\rout$ can move to the inside terminal position only if the yellow discs of radius $1/2$ below it move to the left or to the right.
These yellow discs can move left or right only if at least one of the adjacent radius-$1$ violet discs moves one unit down, which by Lemma~\ref{lem:signal} can happen only if both edge robots $\rin_1$ and $\rin_2$ move to their outside terminal positions. 
In Section~\ref{sec:main} we will prove that the `protected' property of this gadget is preserved.
That is, we will show a correspondence between a valid reconfiguration of the robots and a sequence of edge flips in the graph $G$, such that no two input edges of a \ornode in $G$ are both directed inwards at the same time.
In Section~\ref{sec:correctness} we will prove the following lemma.

\begin{lemma}\label{lem:or-correct}
In the \ornode gadgets, the edge robot `$\rout$' can be moved to the inside terminal position only if at least one of the edge robots `$\rin_1$' and `$\rin_2$' is moved to (or beyond) their outside terminal positions.
The edge robot `$\rin_1$' (`$\rin_2$') can be moved to its inside terminal position only if the edge robot `$\rout$' is moved to (or beyond) its outside terminal position.
\end{lemma}

\subparagraph{Connector gadgets.} Our last type of gadget, the connector gadget, is shown in Figure~\ref{fig:connector_gadgets}.
The two versions of the connector gadget represent a corner piece of an edge (Figure~\ref{fig:connector_gadgets} (left)) and a straight piece of an edge (Figure~\ref{fig:connector_gadgets} (right)).
The connector gadgets consist of a chain of parallel and/or perpendicular components, possibly extended by a sequence of aligned unit~discs. The following lemma follows directly from Lemma~\ref{lem:signal}.
\begin{lemma}\label{lem:connector}
In the connector gadget, an edge robot can only be moved to its inside terminal position if the other edge robot is moved to (or beyond) its outside terminal position.
\end{lemma}

%

\medskip

Using the three described types of gadgets, we now build a complete $2$-DRMP instance corresponding to the NCL machine on a constraint graph $G$.
We fill the cells of the grid-like environment, dual to an orthogonal drawing of $G$, with the \andnode, \ornode, and connector gadgets.
Figure~\ref{fig:gadget_graph} shows an example of a $2$-DRMP instance constructed for the constraint graph from Figure~\ref{fig:grids}.
\begin{figure}[t]
    \centering
    \includegraphics{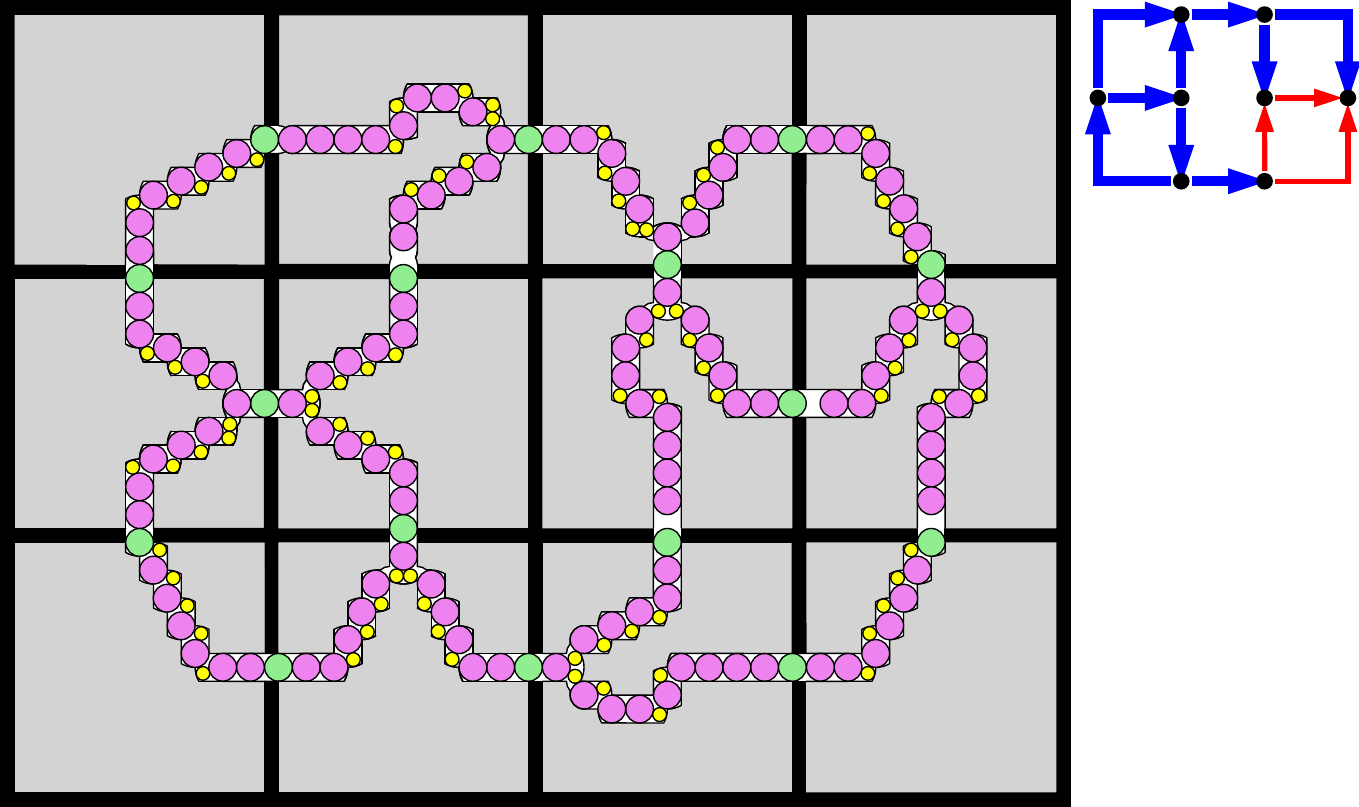}
    \caption{An instance of $k$-DRMP built for the NCL constraint graph $G$ from Figure~\ref{fig:grids}.}
    \label{fig:gadget_graph}
\end{figure}
The construction of the $2$-DRMP instance, and Lemmas~\ref{lem:and-correct},~\ref{lem:or-correct}, and~\ref{lem:connector}, imply the main results of this paper, which we formally prove in the next section.


\begin{restatable}{theorem}{thmmultitomultiarc}
The multi-to-multi, multi-to-single, and multi-to-single-in-class $k$-DRMP problems are \pspace-hard for two classes of unlabeled disc robots moving amidst obstacles constructed out of line segments and circular arcs.
\end{restatable}

Finally, we argue that the circular arcs in the construction can be approximated with circumscribed polygonal chains without changing the validity of the reduction.

\begin{restatable}{theorem}{thmmultitomulti}
The multi-to-multi, multi-to-single, and multi-to-single-in-class $k$-DRMP problems are \pspace-hard for two classes of unlabeled disc robots moving in a polygonal environment.
\end{restatable}

\subsection{Correctness of the gadgets}\label{sec:correctness}
We now prove that the gadgets described in the previous section indeed correspond to their respective nodes in a constraint graph.
Let us first take a closer look at the parallel and perpendicular components that make up our gadgets.
\begin{figure}[t]
    \centering
%
\begin{subfigure}[t]{.4\textwidth}
\includegraphics[page=6]{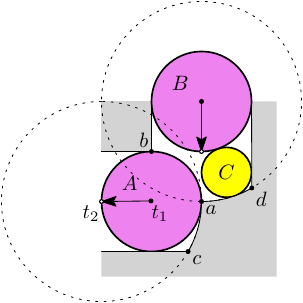}
\caption{Perpendicular component}
\label{fig:default_corner_perp}
\end{subfigure}
    \hfil
%
\begin{subfigure}[t]{.4\textwidth}
\includegraphics[page=3]{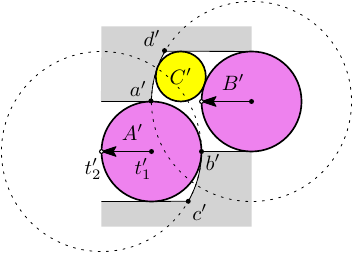}
\caption{Parallel component}
\label{fig:default_corner_par}
\end{subfigure}
    \caption{Detailed design of the perpendicular (a) and parallel (b) components which are present in every gadget. The two possible components of corner situations inside a gadget.}
    \label{fig:default_corner}
\end{figure}

\subparagraph{Properties of parallel and perpendicular components.}
Figure \ref{fig:default_corner} shows a detailed design of the parallel (left) and perpendicular (right) components.
Discs $A$ and $B$ have radius $1$, disc $C$ has radius $1/2$, and the dashed circles have radius $2$.
These configurations can be mirrored and rotated by $90^\circ$, $180^\circ$ and $270^\circ$.
We will chain the components to enforce the terminal positions of the edge robots.
The red arrows indicate movement of the unit discs in the sketched situation.
The dashed circles indicate the circular segments of the boundary of the free space.
The white dots indicate terminal positions of $A$ and $B$.

We first look at the definition of the perpendicular component as shown in Figure~\ref{fig:default_corner} (left).
Let point $a$ be at $(0,0)$. Then points $b=(-1,1)$, $c=(\sqrt{3}-2, -1)$, and $d=(1,2-\sqrt{3})$.
Points $t_1=(-1,0)$ and $t_2=(-2,0)$ are terminal positions of disc $A$.
Indeed, by construction the center of disc $A$ must lie on the closed segment $t_1 t_2$.
In the configuration of the discs depicted in the figure, disc $A$ is centered at $t_1$, disc $B$ is centered at $(0,2)$, and disc $C$ is centered at $(0.5, 2-\sqrt{2})$.
The circular segments $ac$ and $ad$ are arcs of $30^\circ$ of radius-$2$ circles centered at $t_2$ and $(0,2)$ respectively.

Now we define the parallel component as shown in Figure \ref{fig:default_corner} (right).
Let point $a'$ be at $(0,0)$.
Then $b'=(1,-1)$, $c'=(\sqrt{3}-1,-2)$, and $d'=(2-\sqrt{3},1)$.
Points $t'_1=(0,-1)$ and $t'_2=(-1,-1)$ are terminal positions of disc $A'$.
In the example depicted in the figure, disc $A'$ is centered at $t'_1$, disc $B'$ is centered at $(2,0)$, and disc $C'$ is centered at  $(2-\sqrt{2},0.5)$.
The circular segments $b'c'$ and $a'd'$ are arcs of $30^\circ$ of the radius-$2$ circles centered at $t'_2$ and $(2,0)$ respectively.

\begin{figure}[t]
    \centering
    \begin{subfigure}{.4\textwidth}
    \centering
%
\includegraphics[page=7]{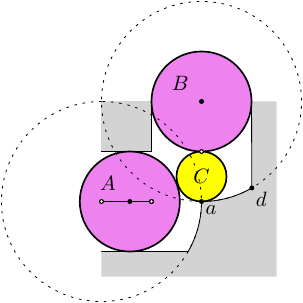}
    \caption{Small disc on the corner and on the same axis as B}
    \label{fig:default_corner_hor_on}
    \end{subfigure}
    \hfil
    \begin{subfigure}{.4\textwidth}
    \centering
%
\includegraphics[page=8]{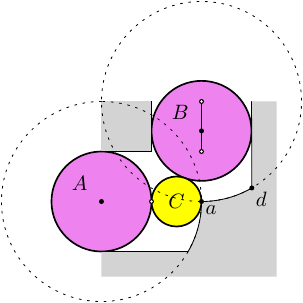}
    \caption{Small disc on the corner and on the same axis as A}
    \label{fig:default_corner_hor_term}
    \end{subfigure}
    \caption{Perpendicular movement with the discs shifted}
    \label{fig:shifted_hor_on}
\end{figure}
\begin{figure}[t]
    \centering
    \begin{subfigure}{.4\textwidth}
    \centering
%
\includegraphics[page=4]{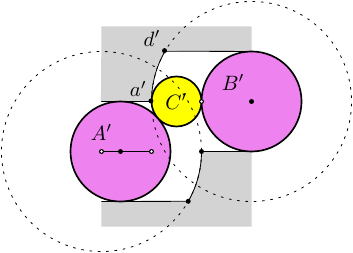}
    \caption{Small disc on the corner and on the same axis as B}
    \label{fig:parallel_component_on}
    \end{subfigure}
    \hfill
    \begin{subfigure}{.4\textwidth}
    \centering
%
\includegraphics[page=5]{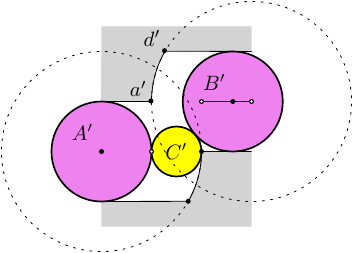}
    \caption{Small disc on the corner and on the same axis as A}
    \label{fig:parallel_component_term}
    \end{subfigure}
    \caption{Perpendicular movement with the discs shifted}
    \label{fig:parallel_component}
\end{figure}

\begin{lemma}\label{lem:perpendicular}
In the perpendicular component, disc $B$ can not move down before disc $A$ has moved by distance of at least $\sqrt{2}-1$ from $t_1$ towards $t_2$.
\end{lemma}
\begin{proof}
Consider the perpendicular configuration depicted in Figure~\ref{fig:default_corner_perp}.
Since the arc $ad$ belongs to a circle of radius $2$, $B$ can not move before the $x$-coordinate of the center of $C$ is smaller than the $x$-coordinate of $a$. 
When the $x$-coordinate of the center of $C$ is equal to the $x$-coordinate of $a$, $B$ is still in its topmost position, but $A$ must have moved by at least $\sqrt{2}-1$ from $t_1$ to prevent overlap with $C$ (see Figure \ref{fig:default_corner_hor_on}).
\end{proof}

\begin{lemma}\label{lem:parallel}
In the parallel component, disc $B'$ can not move left before disc $A'$ has moved by a distance at least $\frac{\sqrt{5}}{2}-\frac{1}{2}$ from $t'_1$ towards $t'_2$.
\end{lemma}
\begin{proof}
Consider the parallel configuration depicted in Figure~\ref{fig:default_corner_par}.
Since the arc $a'd'$ belongs to a circle of radius $2$, $B'$ can not move before the $y$-coordinate of $C'$ becomes smaller than the $y$-coordinate of $a'$. 
When the $y$-coordinate of $C'$ is equal to the $y$-coordinate of $a'$, $B'$ has not been able to move yet, but $A'$ must have moved by at least $\frac{\sqrt{5}}{2}-\frac{1}{2}$ to prevent overlap with $C'$ (see Figure \ref{fig:parallel_component_on}). 
\end{proof}

\begin{proof}[Proof of Lemma~\ref{lem:signal}]
The proof directly follows from Lemmas~\ref{lem:perpendicular} and~\ref{lem:parallel}.
Consider, as an example a chain of a perpendicular and a parallel component as depicted in Figure~\ref{fig:default_stacked}.
By Lemma~\ref{lem:parallel}, disc $B'$ can only move away from its terminal position $t'_2$ when the radius-$(1/2)$ disc $C'$ moves left beyond point $a'$.
Thus, disc $B$ has to move down by at least $\frac{\sqrt{5}}{2}-\frac{1}{2}$ before $B'$ can move.
However, if disc $B$ moves $2-\sqrt{2}$ or more, by Lemma \ref{lem:perpendicular}, $A$ must be in the terminal position $t_2$.
As $2-\sqrt{2} < \frac{\sqrt{5}}{2}-\frac{1}{2}$, we have that either disc $B'$ must be in its terminal position $t'_2$, or disc $A$ must be in its terminal position $t_2$.
If the chain of parallel and perpendicular components is longer, the extremal positions of the radius-$1$ discs beyond either $A$ or $B'$ are enforced.
Thus, if there is a perpendicular and a parallel component in the chain, the lemma holds.
\end{proof}

\subparagraph{Correctness of connector gadgets.}
The correctness of the connector gadgets follows immediately from Lemma~\ref{lem:signal}.
%
\begin{proof}[Proof of Lemma~\ref{lem:connector}]
\begin{figure}[t]
    \centering
%
%
\includegraphics[page=10]{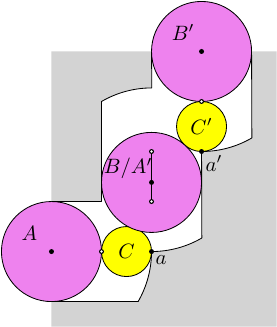}
    \caption{A parallel and a perpendicular component chained. If disc $C$ is horizontally aligned with $a$, and disc $C'$ is vertically aligned with $a'$, then disc $B$ must overlap either $C$ or $C'$.}
    \label{fig:default_stacked}
\end{figure}
As the connector gadgets consist of chains containing a parallel and a perpendicular component, by Lemma~\ref{lem:signal}, the current lemma holds.
\end{proof}

\subparagraph{Correctness of \andnode gadgets.}
%
The \andnode gadgets, as shown in Figure~\ref{fig:and}, use two small disc robots with a radius of $1/2$. 
They have a specific layout of the obstacles, which enable the functionality of the gadgets. 
A close-up view of this layout can be seen in Figure~\ref{fig:AND_closeup}. 
Disc robots $B$ and $B'$ can move along the arc $a'a$ which lies on a circle of radius $2$.
This means that $A$ can not move as long as either disc $B$ or disc $B'$ is on the arc $a'a$.
Arcs $cd$, $c'd'$, $bc$, and $b'c'$ lie on circles of radius $1$.
Thus discs $C$ and $C'$ tightly fit along the respective arcs $bc$ and $b'c'$.
Points $c$ and $c'$ are terminal position of the discs $B$ and $B'$.
Indeed, they cannot move further than $c$ (or $c'$), otherwise they would overlap with disc $C$ (or $C'$).
Thus, discs $B$ and $B'$ cannot both fit in the free space above $C$ or in the free space above $C'$. Therefore, both discs $C$ and $C'$ must move down to enable disc $A$ to move down. If $A$ is able to reach the arc $aa'$, the discs $B$ and $B'$ should have $x$-coordinates at most $a-1/2$ and at least $a'+1/2$ respectively. Then, discs $C$ and $C'$ must move down by at least distance $0.8539$.
A similar argument to the one in Lemmas~\ref{lem:perpendicular} and~\ref{lem:parallel} will force the radius-$1$ discs to be moved to their terminal positions in the direction away from $A$.
%

\begin{figure}[t]
    \centering
    \includegraphics[page=3]{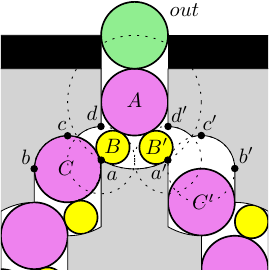}
    \hfil
    \includegraphics[page=4]{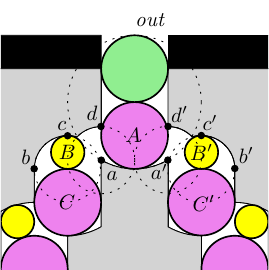}
    \caption{A close-up view of the functionality of the \andnode gadgets, the full gadgets can be seen in Figure~\ref{fig:and}.}
    \label{fig:AND_closeup}
\end{figure}

%
\begin{proof}[Proof of Lemma~\ref{lem:and-correct}]
In Figure \ref{fig:AND_closeup} the \andnode gadget is depicted in more detail.
Since the construction is similar to the connector gadget, we know by Lemma \ref{lem:connector} that discs $C$ and $C'$ from Figure \ref{fig:AND_closeup} can only move when the input edge robots are in their terminal configurations.
By the described construction above, disc $A$ can only move if both $C$ and $C'$ have moved down.
Since $C$ and $C'$ can only move when the inputs are in a terminal configuration, the edge robot $out$ can only move when both edge robots $\rin_1$ and $\rin_2$ are in their outside terminal configuration.
\end{proof}

\subparagraph{Correctness of \ornode gadgets.}
The \ornode gadgets, as shown in Figure~\ref{fig:or}, use two small disc robots with a radius of $1/2$, just like the \andnode gadgets.
However, unlike in the \andnode gadgets, the \ornode gadgets have a different layout of the obstacles, which changes the functionality of the gadgets.
A close-up view of this layout can be seen in Figure~\ref{fig:OR_closeup}. 
Disc robots $B$ and $B'$ can move along the arc $a'a$ which lies on a circle of radius $2$.
As long as $B$ or $B'$ lies on this arc, disc $A$ can not move.
Disc $C$ might move down, which makes space for discs $B$ and $B'$ to move above $C$.
The arc $bc$ lies on a circle of radius $1$.
Discs $B$ and $B'$ have a radius of $1/2$, so they tightly fit along arc $bc$ when $C$ has made space.
Since the arc $cd$ lies on a circle of radius $2$, both $B$ and $B'$ fit above $C$.
When both $B$ and $B'$ have moved into one of the sides, $A$ can move down such that the edge disc can move inside the gadget.
In our construction we forbid the two robots $B$ and $B'$ to separate and move into the free space pockets above the two different radius-$1$ discs.
We can do so because, as we will argue later, for any valid reconfiguration of robots that separates the discs, there will be an equivalent reconfiguration which keeps the discs always together.
%

\begin{figure}[t]
    \centering
    \includegraphics[page=4]{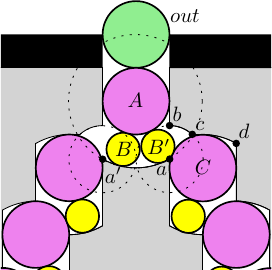}
    \hfil
    \includegraphics[page=3]{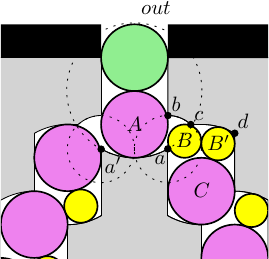}
    \caption{A close-up view of the functionality of the \ornode gadget, the full gadget can be seen in Figure~\ref{fig:or}.}
    \label{fig:OR_closeup}
\end{figure}

%
\begin{proof}[Proof of Lemma~\ref{lem:or-correct}]
The proof directly follows from Lemma~\ref{lem:signal}.
\end{proof}

Observe that, after putting all the gadgets together, all edge robots have a limited set of movements and that inner gadget robots remain in their gadgets.


\begin{observation}\label{lem:terminal}
Each edge robot can be in at most two distinct terminal configurations.
\end{observation}

\subsection{Reduction}\label{sec:main}
We are now ready to prove our main results.

\thmmultitomultiarc*
\begin{proof}
For a given NCL machine built on a constraint graph $G$, consider the corresponding instance of the $k$-DRMP problem constructed as described above.
Furthermore, consider a state $\sigma$ of the NCL machine, and a corresponding $2$-configuration $\S$.
The positions of the edge robots in $\S$ correspond to the orientation of the respective edges of $G$ in $\sigma$.
By Lemmas~\ref{lem:and-correct},~\ref{lem:or-correct}, and~\ref{lem:connector}, if an edge $e\in G$ cannot be flipped, the corresponding inner-gadget radius-$1$ robots are forced to be in one of their terminal positions, and cannot move.
If, however, $e$ can be flipped, some of  the corresponding radius-$1$ robots are able to move between their terminal positions (refer to Figure~\ref{fig:gadget_graph} for an example).
By Lemmas~\ref{lem:and-correct},~\ref{lem:or-correct}, and~\ref{lem:connector}, a flip of the edge $e$ is valid if and only if the corresponding edge robot can move to the opposite terminal position.

Consider the first problem, the \emph{multi-to-multi} $k$-DRMP.
We will show that it is \pspace-hard by a reduction from the \emph{state-to-state} NCL problem.
Recall that the \emph{state-to-state} NCL problem asks whether for a given constraint graph $G$ and for two valid states $\sigma_1$ and $\sigma_2$ of the NCL machine, $\sigma_1$ can be transformed into $\sigma_2$ with edge-flip operations.
From $\S$ and $\T$ we construct two $2$-configurations $\S'$ and $\T'$, such that all the edge robots are in the same positions as in $\S$ and $\T$, all inner radius-$1$ and radius-$(1/2)$ robots are shifted to their terminal positions consistent with the orientation of the corresponding edges in $G$.
We claim that $\sigma_1$ can be transformed into $\sigma_2$ with edge-flip operations if and only if the robots can be reconfigured from $\S'$ to $\T'$.

Assume that there is a sequence of edge flips transforming $\sigma_1$ into $\sigma_2$.
For each flip, by Lemmas~\ref{lem:and-correct},~\ref{lem:or-correct}, and~\ref{lem:connector}, we can reconfigure the robots of the $k$-DRMP instance in correspondence to the changes of orientations of the flipped edges.

It remains to show that if the robots of the $k$-DRMP instance can be reconfigured from $\S'$ to $\T'$, then there is a valid edge-flip sequence transforming $\sigma_1$ into $\sigma_2$.
Consider the reconfiguration over time, and extract the order in which the edge robots reach one of their terminal configurations.
If two edge robots are both in some intermediate positions between their terminal locations, then these edge robots can move independently from one another.
We can modify the reconfiguration schedule such that at each moment in time only one edge robot can be located at an intermediate position between its terminal positions.

We still need to argue that we can preserve the `protected' property of the \ornode gadgets.
Suppose that in the process of reconfiguration, at some moment, two input edge robots are moved to the outside terminal positions.
If one of them, say $\rin_1$, reverts before the robot $\rout$ moves to its inner terminal position, then we simply ignore the move of $\rin_1$ (and the robots in the chain from $\rin_1$ to $out$) outside.
Let robot $\rout$ move to the inner terminal position.
Consider the positions of the discs $B$ and $B'$ (recall Figure~\ref{fig:OR_closeup}).
If they both are located above one radius-$1$ disc in the chain from $\rin_1$ to $out$, then we can change the schedule to stop $\rin_2$ from moving to the outside terminal position.
If $B$ and $B'$ are separated, then we can modify the schedule to move $B$ and $B'$ into the same free-space pocket, and stop the other input edge robot from moving outside.
In all cases, we can modify the schedule to prevent both input edge robots to be in their outside terminal positions.
Thus, we have a reconfiguration schedule which preserves the `protected' property of the \ornode nodes, and has edge robots move between their terminal positions one at a time.

The order in which the edge robots move between their terminal positions gives the order of valid edge flips in $G$.
Indeed, by Lemmas~\ref{lem:and-correct},~\ref{lem:or-correct}, and~\ref{lem:connector}, if an edge robot, corresponding to some edge $e = (u,v)$ $\in G$, can move, the in-flow property of the two corresponding nodes $u$ and $v$ in $G$ is satisfied by the edges other than $e$.
Thus, the multi-to-multi $k$-DRMP problem is \pspace-hard.

Now, consider the multi-to-single and multi-to-single-in-class versions of the $k$-DRMP problem.
By a reduction from state-to-edge NCL problem, we show that these two problems are \pspace-hard.
The argument follows the same lines as for the multi-to-multi case, except that instead of the target $2$-configuration $\T$, we are given a target location for a robot.
We will select the edge robot corresponding to the edge to be flipped in the NCL problem, and specify a proper terminal positions as the target location for the robot.
\end{proof}

\subparagraph{Remark.}
Note, that we can remove the use of circular arcs in our construction.
Consider a small fixed value $\varepsilon>0$.
There exists a value $d(\varepsilon)>0$, such that, if we replace the arcs in the construction with circumscribed polygonal chains with edge length at most $d$, the edge robots will be bound to $\varepsilon$-neighborhoods of the terminal positions.
Indeed, for small enough $\varepsilon$, Lemmas~\ref{lem:and-correct},~\ref{lem:or-correct}, and~\ref{lem:connector} will still hold, with modified statements considering the $\varepsilon$-neighborhoods of the terminal positions instead of simply the terminal positions.
Thus, the following result holds.
\thmmultitomulti*

\section{Conclusion}\label{sec:conclusion}
In this paper we have shown that the three variants of the disc-robot motion planning problem are \pspace-hard, even for two classes of unlabeled disc robots with two different radii, moving in a polygonal environment.
This is a first step towards settling the complexity of unlabeled unit disc robot motion planning.
%
%
Our gadgets do not seem to generalize to a single class of robots.
The complexity of unlabeled unit disc robot motion planning remains an interesting open problem.


\bibliography{discs-pspace-hard}

\end{document}